\documentclass[DIV=12,11pt]{scrartcl}

\usepackage{amsthm,amssymb,amsmath}
\usepackage[noblocks]{authblk}

\setkomafont{title}{\normalfont \LARGE \bfseries}
\setkomafont{section}{\normalfont \Large \bfseries \boldmath}
\setkomafont{subsection}{\normalfont \large \bfseries}
\setkomafont{subsubsection}{\normalfont \normalsize \bfseries}
\setkomafont{descriptionlabel}{\itshape}
\setkomafont{paragraph}{\normalfont \bfseries \boldmath}

\newtheorem{theorem}{Theorem}
\newtheorem{lemma}[theorem]{Lemma}
\newtheorem{corollary}[theorem]{Corollary}

\newtheorem{definition}[theorem]{Definition}

\usepackage{graphicx}
\usepackage{hyperref}
\usepackage[numbers]{natbib}
\usepackage{url}

\newcommand{\Exp}{\mathop{\mathgroup\symoperators Exp}\nolimits}

\newcommand{\blokje}{}

\begin{document}

\title{Probabilistic Analysis of Optimization Problems on Generalized Random Shortest Path Metrics\footnote{An extended abstract of this work has appeared in the \textit{Proceedings of the 13th International Conference and Workshops on Algorithms and Computation (WALCOM 2019)}.}}

\author{Stefan Klootwijk}
\author{Bodo Manthey}
\author{Sander K. Visser}

\affil{University of Twente, Enschede, The Netherlands, \texttt{\{s.klootwijk,b.manthey\}@utwente.nl, s.k.visser@alumnus.utwente.nl}}

\maketitle

\begin{abstract}
Simple heuristics often show a remarkable performance in practice for optimization problems. Worst-case analysis often falls short of explaining this performance. Because of this, ``beyond worst-case analysis'' of algorithms has recently gained a lot of attention, including probabilistic analysis of algorithms.

The instances of many optimization problems are essentially a discrete metric space. Probabilistic analysis for such metric optimization problems has nevertheless mostly been conducted on instances drawn from Euclidean space, which provides a structure that is usually heavily exploited in the analysis. However, most instances from practice are not Euclidean.
Little work has been done on metric instances drawn from other, more realistic, distributions. Some initial results have been obtained by Bringmann et al.\ (\emph{Algorithmica}, 2013), who have used random shortest path metrics on complete graphs to analyze heuristics.

The goal of this paper is to generalize these findings to non-complete graphs, especially Erd\H{o}s--R\'enyi random graphs. A random shortest path metric is constructed by drawing independent random edge weights for each edge in the graph and setting the distance between every pair of vertices to the length of a shortest path between them with respect to the drawn weights. For such instances, we prove that the greedy heuristic for the minimum distance maximum matching problem, the nearest neighbor and insertion heuristics for the traveling salesman problem, and a trivial heuristic for the $k$-median problem all achieve a constant expected approximation ratio. Additionally, we show a polynomial upper bound for the expected number of iterations of the 2-opt heuristic for the traveling salesman problem.

\end{abstract}

\nocite{Bhamidi2011}
\nocite{Bringmann2015}
\nocite{Byrka2015}
\nocite{Davis1993}
\nocite{Hammersley1965}
\nocite{Hassin1985}
\nocite{Howard2004}
\nocite{Janson1999}
\nocite{Janson2018}
\nocite{Karp1985}
\nocite{Mitzenmacher2005}
\nocite{Reingold1981}
\nocite{Rosenkrantz1977}
\nocite{Ross2010}

\section{Introduction}

Large-scale optimization problems, such as the traveling salesman problem (TSP), show up in many applications. These problems are often computationally intractable. However, in practice often ad-hoc heuristics are successfully used that provide solutions that come quite close to optimal solutions. In many cases these, often simple, heuristics show a remarkable performance, even though the theoretical results about those heuristics are way more pessimistic.

In order to explain this difference, probabilistic analysis has been widely used over the last decades. However, the challenge in probabilistic analysis is to come up with a good probabilistic model: it should reflect realistic instances, but also be sufficiently simple to make the analysis tractable.

So far, in almost all cases, either Euclidean space has been used to generate instances of metric optimization problems, or independent, identically distributed edge lengths have been used. However, both approaches have considerable shortcomings to explain the average-case performance of heuristics on general metric instances: the structure of Euclidean space is heavily used in the probabilistic analysis, but realistic instances are often not Euclidean. The independent, identically distributed edge lengths do not even yield a metric in the first place. In order to overcome these shortcomings, Bringmann et al.~\cite{Bringmann2015} have proposed and analyzed the following model to generate random metric spaces, which had already been proposed by Karp and Steele in 1985 \cite{Karp1985}: given an undirected complete graph, start by drawing random edge weights for each edge independently and then define the distance between any two vertices as the total weight of the shortest path between them, measured with respect to the random weights.

\subsection{Related Work}

Bringmann et al.\ called the model described above \emph{random shortest path metrics}. This model is also known as \emph{first-passage percolation}, introduced by Hammersley and Welsh as a model for fluid flow through a (random) porous medium \cite{Hammersley1965,Howard2004}.

For first passage percolation in complete graphs, the expected distance between two fixed vertices is approximately $\ln(n)/n$ and the expected distance from a fixed vertex to the vertex that is most distant is approximately $2\ln(n)/n$~\cite{Bringmann2015,Janson1999}. Furthermore, the expected diameter of the metric is approximately $3\ln(n)/n$~\cite{Hassin1985,Janson1999}. There are also some known structural properties of first passage percolation on the Erd\H{o}s--R\'enyi random graph. Bhamidi et al.~\cite{Bhamidi2011} have shown asymptotics for both the minimal weight of the path between uniformly chosen vertices in the giant component and for the hopcount, the number of edges, on this path.

Bringmann et al.~\cite{Bringmann2015} used this model on the complete graph to analyze heuristics for matching, TSP, and $k$-median.

\subsection{Our Results}

As far as we know, no heuristics have been studied in this model for non-complete graphs yet.
However, we believe that random shortest path metrics on non-complete graphs will bring us a step further in the direction of realistic input model.

This paper provides a probabilistic analysis of some simple heuristics in the model of random shortest path metrics on non-complete graphs. First, we provide some structural properties of generalized random shortest path metrics (Sect.~\ref{sect:structural}), which can be seen as a generalization of the structural properties found by Bringmann et al.~\cite{Bringmann2015}. Although this generalization might seem straightforward at first sight, it brings up some new difficulties that need to be overcome. Most notably, since we do not restrict ourselves to the complete graph, we cannot make use anymore of its symmetry and regularity. This problem is partially solved by introducing two graph parameters, which we call the cut parameters of a graph (Def.~\ref{def:ab}).

Then, we use these structural insights to perform a probabilistic analysis for some simple heuristics for combinatorial optimization problems (Sect.~\ref{sect:heuristics}), where the results are still depending on the cut parameters of a graph. Finally, we use these results, to show our main results, namely that these simple heuristics achieve constant expected approximation ratios for random shortest path metrics applied to Erd\H{o}s--R\'enyi random graphs (Sect.~\ref{sect:ERRG}).


\section{Notation and Model}

We use $X\sim P$ to denote that a random variable $X$ is distributed using a probability distribution $P$. $\Exp(\lambda)$ is being used to denote the exponential distribution with parameter $\lambda$. In particular, we use $X\sim\sum_{i=1}^n\Exp(\lambda_i)$ to denote that $X$ is the sum of $n$ independent exponentially distributed random variables having parameters $\lambda_1,\ldots,\lambda_n$.

For $n\in\mathbb{N}$, we use $[n]$ as shorthand notation for $\{1,\ldots,n\}$. We denote the $n$th harmonic number by $H_n=\sum_{i=1}^n1/i$. Sometimes we use $\exp$ to denote the exponential function. Finally, if a random variable $X$ is stochastically dominated by a random variable $Y$, i.e., we have $F_X(x)\geq F_Y(x)$ for all $x$ (where $X\sim F_X$ and $Y\sim F_Y$), we denote this by $X\precsim Y$.

\paragraph{Generalized Random Shortest Path Metrics.}

Given an undirected graph $G=(V,E)$ on $n$ vertices, we construct the corresponding generalized random shortest path metric as follows. First, for each edge $e\in E$, we draw a random edge weight $w(e)$ independently from an exponential distribution\footnote{Exponential distributions are technically easiest to handle due to their memorylessness property. 
A (continuous, non-negative) probability distribution of a random variable $X$ is said to be memoryless if and only if $\mathbb{P}(X>s+t\mid X>t)=\mathbb{P}(X>s)$ for all $s,t\geq0$.~\cite[p.~294]{Ross2010}} with parameter 1. Second, we define the distances $d:V\times V\to\mathbb{R}_{\geq0}\cup\{\infty\}$ as follows: for every $u,v\in V$, $d(u,v)$ denotes the length of the shortest $u,v$-path with respect to the drawn edge weights. If no such path exists, we set $d(u,v)=\infty$. By doing so, the distance function $d$ satisfies $d(v,v)=0$ for all $v\in V$, $d(u,v)=d(v,u)$ for all $u,v\in V$, and $d(u,v)\leq d(u,s)+d(s,v)$ for all $u,s,v\in V$. We call the complete graph with distances $d$ obtained from this process a generalized random shortest path metric. If $G=K_n$ (the complete graph on $n$ vertices), then this generalized random shortest path metric is equivalent to the random shortest path metric as defined by Bringmann et al.~\cite{Bringmann2015}

We use the following notation within generalized random shortest path metrics: $\Delta_{\max}:=\max_{u,v}d(u,v)$ denotes the diameter of the graph. Note that $\Delta_{\max}<\infty$ if and only if $G$ is connected. $B_\Delta(v):=\{u\in V\mid d(u,v)\leq\Delta\}$ denotes the `ball' of radius $\Delta$ around $v$, i.e., the set containing all vertices at distance at most $\Delta$ from $v$. $\tau_k(v):=\min\{\Delta\mid|B_\Delta(v)|\geq k\}$ denotes the distance to the $k$th closest vertex from $v$ (including $v$ itself). Equivalently, one can also say that $\tau_k(v)$ is equal to the smallest $\Delta$ such that the ball of radius $\Delta$ around $v$ contains at least $k$ vertices.

Now, $B_{\tau_k(v)}(v)$ denotes the set of the $k$ closest vertices to $v$. During our analysis, we will make use of the size of the cut induced by this set, which we will denote by $\chi_k(v):=|\delta(B_{\tau_k(v)}(v))|$, where $\delta(U)$ denotes the cut induced by $U$.

\paragraph{Erd\H{o}s--R\'enyi Random Graphs.}

The main results of this work consider random shortest path metrics applied to Erd\H{o}s--R\'enyi random graphs. An undirected graph $G(n,p):=G=(V,E)$ generated by this model has $n$ vertices ($V=\{1,\ldots,n\}$) and between each pair of vertices an edge is included with probability $p$, independent of every other pair.

Working with the Erd\H{o}s--R\'enyi random graph introduces an extra amount of stochasticity to the probabilistic analysis, since both the graph and the edge weights are random. In order to avoid this extra stochasticity as long as possible, in Sections \ref{sect:structural} and \ref{sect:heuristics} we start our analysis using an arbitrary fixed (deterministic) graph $G$. Later on, in Section \ref{sect:ERRG} we will consider Erd\H{o}s--R\'enyi random graphs again.

\section{Structural properties}\label{sect:structural}

In order to analyze the structural properties of generalized random shortest path metrics, we first introduce the notion of what we call the cut parameters of a simple graph $G$.
\begin{definition}\label{def:ab}
Let $G=(V,E)$ be a finite simple connected graph. Then we define the cut parameters of $G$ by
\begin{equation*}
\alpha:=\min_{\varnothing\neq U\subset V}\frac{|\delta(U)|}{\mu_U}\qquad\text{and}\qquad\beta:=\max_{\varnothing\neq U\subset V}\frac{|\delta(U)|}{\mu_U},
\end{equation*}
where $\mu_U:=|U|\cdot(|V|-|U|)$ is the maximum number of possible edges in the cut defined by $U$.
\end{definition}
It follows immediately from this definition that $0<\alpha\leq\beta\leq1$ for any finite simple connected graph $G$. Moreover, for any such graph the following holds for all $\varnothing\neq U\subset V$: $\alpha\cdot\mu_U\leq|\delta(U)|\leq\beta\cdot\mu_U$. We observe that the cut parameters of the complete graph are given by $\alpha=\beta=1$.

\paragraph{Distribution of $\tau_k(v)$.}
Now we have a look at the distribution of $\tau_k(v)$. For this purpose we use an arbitrary fixed undirected connected simple graph $G$ (on $n$ vertices) and let $\alpha$ and $\beta$ denote its cut parameters.

The values of $\tau_k(v)$ are then generated by a birth process as follows. (Amongst others, a variant of this process for complete graphs has been analyzed by Davis and Prieditis~\cite{Davis1993} and Bringmann et al.~\cite{Bringmann2015}.) For $k=1$, we have $\tau_k(v)=0$. For $k\geq2$, we look at all edges $(u,x)$ with $u\in B_{\tau_{k-1}(v)}(v)$ and $x\not\in B_{\tau_{k-1}(v)}(v)$. By definition there are $\chi_{k-1}(v)$ such edges. Moreover the length of these edges is conditioned to be at least $\tau_{k-1}(v)-d(v,u)$. Using the memorylessness of the exponential distribution, we can now see that $\tau_k(v)-\tau_{k-1}(v)$ is the minimum of $\chi_{k-1}(v)$ (standard) exponential variables, or, equivalently, $\tau_k(v)-\tau_{k-1}(v)\sim\Exp(\chi_{k-1}(v))$. We use this result to find bounds for the distribution of $\tau_k(v)$.
\begin{lemma}
	\label{lemma:X}
	For all $k\in[n]$ and $v\in V$ we have,
	\begin{align*}
	\alpha k(n-k) \leq \chi_k(v) \leq \beta k(n-k).
	\end{align*}
\end{lemma}
\begin{proof}
	By definition, $\chi_k(v)$ is the size of a cut induced by a set of $k$ vertices. The result follows immediately since $\alpha$ and $\beta$ are the cut parameters of $G$.\blokje
\end{proof}
\begin{lemma}
	\label{lemma:tau}
	For all $k\in[n]$ and $v\in V$ we have,
	\begin{align*}
	\sum_{i=1}^{k-1}\Exp(\beta i(n-i)) \precsim \tau_k(v) \precsim \sum_{i=1}^{k-1}\Exp(\alpha i(n-i)).
	\end{align*}
\end{lemma}
\begin{proof}
	As previously stated, $\tau_{i}(v)-\tau_{i-1}(v) \sim \Exp(\chi_{i-1}(v))$. Inductively, we obtain that
	\begin{align*}
	\tau_{k}(v)\sim \sum_{i=1}^{k-1} \Exp(\chi_i(v)).
	\end{align*}
	Using the result of Lemma \ref{lemma:X}, we can bound this distribution to obtain the desired result.\blokje
\end{proof}
Exploiting the linearity of expectation, the fact that the expected value of an exponentially distributed random variable with parameter $\lambda$ is $1/\lambda$ and the fact that $\sum_{i=1}^{k-1}1/(i(n-i))=(H_{k-1}+H_{n-1}-H_{n-k})/n$, we obtain the following corollary.
\begin{corollary}
\label{corollary:harmonic}
For all $k\in[n]$ and $v\in V$ we have,
\begin{equation*}
\frac{H_{k-1}+H_{n-1}-H_{n-k}}{\beta n}\leq\mathbb{E}(\tau_k(v))\leq\frac{H_{k-1}+H_{n-1}-H_{n-k}}{\alpha n}.
\end{equation*}
\end{corollary}
From this result, we can derive the following extensions of two known results. First of all, if we randomly pick two vertices $u,v\in V$, then averaging over $k$ yields that the expected distance $\mathbb{E}[d(u,v)]$ between them is bounded between $\frac{H_{n-1}}{\beta(n-1)}\approx\ln(n)/\beta n$ and $\frac{H_{n-1}}{\alpha(n-1)}\approx\ln(n)/\alpha n$, which is in line with the known result for complete graphs, where we have $\mathbb{E}[d(u,v)]\approx\ln(n)/n$~\cite{Bringmann2015,Davis1993,Janson1999}.  Secondly, for any vertex $v$, the longest distance from it to another vertex is $\tau_n(v)$, which in expectation is bounded between $\frac{2H_{n-1}}{\beta n}\approx2\ln(n)/\beta n$ and $\frac{2H_{n-1}}{\alpha n}\approx2\ln(n)/\alpha n$, which also is in line with the known result for complete graphs, where we have an expected value of approximately $2\ln(n)/n$~\cite{Bringmann2015,Janson1999}.

It is also possible to find bounds for the cumulative distribution function of $\tau_k(v)$. To do so, we define $F_k(x)=\mathbb{P}(\tau_k(v)\leq x)$  for some fixed vertex $v\in V$.
\begin{lemma}
	\label{technical:exp}
	\emph{\cite[Lemma~3.2]{Bringmann2015}}
	Let $X \sim \sum_{i=1}^n \Exp(ci)$. Then, for any $a\geq 0$ we have $\mathbb{P}(X\leq a)=\left(1-e^{-ca}\right)^n$.
\end{lemma}
\begin{lemma}
	\label{lemma:F}
	For all $x\geq0$ and $k\in[n]$ we have,
	\begin{equation*}
	\left(1-\exp(-\alpha(n-k)x)\right)^{k-1} \leq F_k(x) \leq \left(1-\exp(-\beta nx)\right)^{k-1}.
	\end{equation*}
\end{lemma}
\begin{proof}
	By Lemma \ref{lemma:tau} we have
	\begin{equation*}
	\sum_{i=1}^{k-1}\Exp\left(\beta i(n-i)\right) \precsim \tau_k(v) \precsim \sum_{i=1}^{k-1} \Exp\left(\alpha i(n-i)\right).
	\end{equation*}
	Since $ni\geq i(n-i)\geq (n-k)i$ for all $i\in[k-1]$, we have $\Exp(\beta ni)\precsim\Exp(\beta i(n-i))$ and $\Exp(\alpha i(n-i))\precsim\Exp(\alpha(n-k)i)$ for all $i\in[k-1]$, from which we obtain that
	\begin{equation*}
	\sum_{i=1}^{k-1}\Exp\left(\beta ni)\right) \precsim \tau_k(v) \precsim \sum_{i=1}^{k-1} \Exp\left(\alpha(n-k)i\right).
	\end{equation*}
	Combining this with the definition of stochastic dominance and with Lemma \ref{technical:exp}, gives the desired result.\blokje
\end{proof}
We can improve this result slightly.
\begin{lemma}
    \label{lemma:distfunction}
    For all $x\geq0$ and $k\in[n]$ we have,
    \begin{equation*}
    F_k(x) \geq\left(1-\exp(-\alpha nx/4)\right)^{n}.
    \end{equation*}
\end{lemma}
\begin{proof}
	Note that $\tau_k(v)$ is monotonically increasing in $k$. This implies $F_{k+1}(x) \leq F_k(x)$, so we only need to prove our claim for the case $k=n$. In this case, by Lemma \ref{lemma:tau}, we have $\tau_n(v) \precsim \sum_{i=1}^{n-1}\Exp\left(\lambda_i\right)$ with $\lambda_i:=\alpha i(n-i)=\lambda_{n-i}$. Exploiting the symmetry around $n/2$, we obtain
	\begin{equation*}
	\tau_n(v) \precsim \sum_{i=1}^{\lfloor n/2\rfloor}\Exp(\lambda_i)+\sum_{i=1}^{\lfloor n/2\rfloor}\Exp(\lambda_i).
	\end{equation*}
	This enables us to find a lower bound for $F_n(x)$ as follows:
	\begin{align*}
	F_n(x)=\mathbb{P}(\tau_n(v) \leq x)&\geq\mathbb{P}\left(\sum_{i=1}^{\lfloor n/2\rfloor}\Exp(\lambda_i)+\sum_{i=1}^{\lfloor n/2\rfloor}\Exp(\lambda_i)\leq x\right)\\
	&\geq\mathbb{P}\left(\sum_{i=1}^{\lfloor n/2\rfloor}\Exp(\lambda_i)\leq x/2\right)^2.
	\end{align*}
	Since $i(n-i)\geq i\lceil n/2\rceil$ for all $i\in[\lfloor n/2\rfloor]$, we have $\Exp(\lambda_i)\precsim\Exp(\alpha i\lceil n/2\rceil)$. Combining this with Lemma \ref{technical:exp} yields
	\begin{equation*}
	F_n(x)\geq\mathbb{P}\left(\sum_{i=1}^{\lfloor n/2\rfloor}\Exp(\alpha i\lceil n/2\rceil)\leq x/2\right)^2=\left(1-\exp(-\alpha x\lceil n/2\rceil/2)\right)^{2\lfloor n/2\rfloor}.
	\end{equation*}
	Using the inequalities $\lceil n/2\rceil\geq n/2$ and $2\lfloor n/2\rfloor\leq n$ we end up with the desired result.\blokje
\end{proof}
Using this improved bound for the cumulative distribution function of $\tau_k(v)$, we can derive the following tail bound for the diameter $\Delta_{\max}$.
\begin{lemma}
\label{lemma:diameter}
	Define $\Delta_{\max}=\max_{u,v\in V}\{d(u,v)\}$. For any fixed $c$ we have 
	\begin{equation*}
	\mathbb{P}( \Delta_{\max} > c \ln(n)/\alpha n) \leq n^{2-c/4}.
	\end{equation*}
\end{lemma}
\begin{proof}
	Clearly, we have $\Delta_{\max}=\max_v\tau_n(v)$. For $v\in V$, let $\mathcal{E}_v$ denote the event that $\tau_n(v)>c\ln(n)/\alpha n$. From Lemma \ref{lemma:distfunction} we know that $\mathbb{P}(\mathcal{E}_v)=1-F_n(c\ln(n)/\alpha n)\leq1-(1-\exp(-c\ln(n)/4))^n$. Combining this with a union bound, we can derive that
	\begin{equation*}
	\mathbb{P}\left(\Delta_{\max}>\frac{c\ln(n)}{\alpha n}\right)\leq\sum_{v\in V}\mathbb{P}(\mathcal{E}_v)\leq n\cdot\left(1-\left(1-n^{-c/4}\right)^n\right)\leq n^{2-c/4},
	\end{equation*}
	where the last inequality can be derived using Bernoulli's inequality.\blokje
\end{proof}

\paragraph{Clustering.} 
In this section we show that we can partition the vertices of generalized random shortest path metrics into a small number of clusters with a given maximum diameter. Before we prove this main result, we first provide a tail bound for $|B_\Delta(v)|$.
\begin{lemma}
\label{lemma:ball}
For $n\geq5$ and for any fixed $\Delta\geq0$ we have,
\begin{equation*}
\mathbb{P}\left(|B_\Delta(v)| < \min \left\{\exp(\alpha \Delta n/5 ), \frac{n+1}{2}\right\}\right) \leq \exp(-\alpha \Delta n/5).
\end{equation*}
\end{lemma}
\begin{proof}
	We have $|B_\Delta (v)|\geq k$ if and only if $\tau_k(v) \leq \Delta$. Using Lemma \ref{lemma:F}, we obtain
	\begin{align*}
	&\mathbb{P}\left(|B_\Delta(v)| < \min \left\{\exp(\alpha \Delta (n-1)/4 ), \frac{n+1}{2}\right\}\right)\\
	&\quad\leq1-\left(1-\exp\left(-\alpha\Delta\left(n-\frac{n+1}2\right)\right)\right)^{\exp(\alpha\Delta(n-1)/4)-1}\\
	&\quad\leq1-\left(1-\exp\left(-\alpha\Delta(n-1)/2\right)\right)^{\exp(\alpha\Delta(n-1)/4)}\\
	&\quad\leq\exp(-\alpha\Delta(n-1)/4),
	\end{align*}
	where the last inequality can be derived using Bernoulli's inequality. Using $(n-1)/4\geq n/5$ for $n\geq5$ finishes the proof.\blokje
\end{proof}
We use the result of this lemma to prove our main structural property for generalized random shortest path metrics. 
\begin{theorem}
\label{theorem:cluster}
For any fixed $\Delta\geq0$, if we partition the vertices into clusters, each of diameter at most $4\Delta$, then the expected number of clusters needed is bounded from above by $O(1+n/\exp( \alpha \Delta n/5))$.
\end{theorem}
\begin{proof}
	Define $s_\Delta = \min\{\exp( \alpha\Delta  n/5, (n+1)/2\}$. We call vertex $v$ \emph{$\Delta$-dense} if $|B_\Delta(v)|\geq s_\Delta$ and \emph{$\Delta$-sparse} otherwise. In both cases we call the set $B_\Delta(v)$ of vertices within distance $\Delta$ of $v$ the \emph{$\Delta$-ball} of $v$. By Lemma \ref{lemma:ball} we can bound the expected number of $\Delta$-sparse vertices by $O(n/s_\Delta)$. We put each $\Delta$-sparse vertex in its own cluster (of size 1), which has diameter $0\leq4\Delta$.
	
	This leaves us with the $\Delta$-dense vertices. We cluster them according to the following process. Consider an auxiliary graph $H$ whose vertices are the $\Delta$-dense vertices and where two vertices are connected by an edge if and only if their corresponding $\Delta$-balls are not disjoint. Now, consider an arbitrary maximal independent set $S$ in $H$. Since $|B_\Delta(v)|\geq s_\Delta$ and $B_\Delta(u)\cap B_\Delta(v)=\varnothing$ for any $u,v\in S$, it follows that $|S|\leq n/s_\Delta$. Now, we form the initial clusters $C_1,\ldots,C_{|S|}$ each of which is equal to the $\Delta$-ball corresponding to one of the vertices in $S$. Observe that these initial clusters have diameter at most $2\Delta$.
	
	Now consider an arbitrary $\Delta$-dense vertex $v$ that is not part of any cluster yet. Since $S$ is a maximal independent set, we know that there exists a $u\in S$ such that $B_\Delta(u)\cap B_\Delta(v)\neq\varnothing$. We add $v$ to the cluster that contains $u$. If we take $x\in B_\Delta(u)\cap B_\Delta(v)$, then we can see that $d(v,u)\leq d(v,x)+d(x,u)\leq\Delta+\Delta=2\Delta$. We repeat this step until all $\Delta$-dense vertices have been added to some initial cluster. By construction, the diameter of each cluster is at most $4\Delta$ after this process: consider any vertices $x,y$ in the same cluster, that originally corresponded to a vertex $u\in S$. Then we have $d(x,y)\leq d(x,u)+d(u,y)\leq2\Delta+2\Delta=4\Delta$.
	
	So, now we have in expectation $O(n/s_\Delta)$ clusters each containing one $\Delta$-sparse vertex, and at most $n/s_\Delta$ clusters each containing at least $s_\Delta$ $\Delta$-dense vertices, all with diameter at most $4\Delta$. The total number of clusters is $O(n/s_\Delta)=O(1+n/\exp( \alpha \Delta n/5))$.\blokje
\end{proof}

\section{Analysis of Heuristics}\label{sect:heuristics}

In this section we bound the expected approximation ratios of the greedy heuristic for minimum-distance perfect matching, the nearest neighbor and insertion heuristics for the traveling salesman problem, and a trivial heuristic for the $k$-median problem. For this purpose we still use an arbitrary fixed undirected connected simple graph $G$ (on $n$ vertices) and let $\alpha$ and $\beta$ denote its cut parameters. The results in this section will depend on $\alpha$ and $\beta$.

\paragraph{Greedy Heuristic for Minimum-Distance Perfect Matching.}
The mini-mum-distance perfect matching problem has been widely analyzed throughout history. We do for instance know that the worst-case running-time for finding a minimum distance perfect matching is $O(n^3)$, which is high when considering a large number of vertices. Because of this, simple heuristics are often used, with the greedy heuristic probably being the simplest of them: at each step, add a pair of unmatched vertices to the matching such that the distance between the added pair of vertices is minimized. From now on, let $\mathsf{GR}$ denote the cost of the matching computed by this heuristic and let $\mathsf{MM}$ denote the value of an optimal matching.

The worst-case approximation ratio of this heuristic on metric instances is known to be $O(n^{\log_2(3/2)})$~\cite{Reingold1981}. Furthermore, for random shortest path metrics on complete graphs (for which the cut parameters are given by $\alpha=\beta=1$) the heuristic has an expected approximation ratio of $O(1)$~\cite{Bringmann2015}. We extend this last result to general values for $\alpha$ and $\beta$ and show that the greedy matching heuristic has an expected approximation ratio of $O(\beta/\alpha)$.
\begin{theorem}
\label{theorem:greedy}
$\mathbb{E}[\mathsf{GR}]=O\left(1/\alpha\right)$.
\end{theorem}
\begin{proof}
	Let $\Delta_i:=i/\alpha n$. We divide the run of the greedy heuristic in phases as follows: the algorithm is in phase $i$ if a pair $(u,v)$ is added to the matching such that $d(u,v) \in(4\Delta_{i-1},4\Delta_i]$. Using Lemma \ref{lemma:diameter}, we can show that the expected sum of all distances greater than or equal to $\Delta_{\omega(\ln(n))}$ is $o(1/\alpha)$, so we can ignore the corresponding phases in our analysis.
	
	We now estimate the contribution of the other phases to the greedy matching. By Theorem \ref{theorem:cluster}, after phase $i-1$, we can partition the vertices in an expected number of $O(1+n/\exp((i-1)/5))$ clusters, each of diameter at most $4\Delta_{i-1}$. Each such cluster can have at most one unmatched vertex. So, after phase $i-1$ there are at most $O(1+n/\exp((i-1)/5))$ unmatched vertices left. Therefore, in expectation at most $O(1+n/\exp((i-1)/5))$ pairs of unmatched vertices can be added in phase $i$, each contributing a distance of at most $4\Delta_i$. So, the total contribution of phase $i$ is in expectation at most $O(\frac{i}{\alpha n}(1+n/\exp((i-1)/5)))$. Summing over all phases yields
	\begin{equation*}
	\mathbb{E}[\mathsf{GR}]=o\left(\frac1\alpha\right)+\sum_{i=1}^{O(\ln(n))}O\left(\frac1\alpha\left(\frac{i}{n}+\frac{i}{e^{(i-1)/5}}\right)\right)= o\left(\frac1\alpha\right)+O\left(\frac1\alpha\right)=O\left(\frac1\alpha\right),
	\end{equation*}
	which completes the proof.\blokje
\end{proof}
\begin{lemma}
\label{technical:Janson}
\emph{\cite[Thm.~5.1(iii)]{Janson2018}}
Let $X\sim\sum_{i=1}^nX_i$ with $X_i\sim\Exp(a_i)$ independent. Let $\mu=\mathbb{E}[X]=\sum_{i=1}^n(1/a_i)$ and $a_*=\min_ia_i$. For any $\lambda\leq1$,
\begin{equation*}
\mathbb{P}(X\leq\lambda\mu)\leq\exp(-a_*\mu(\lambda-1-\ln(\lambda))).
\end{equation*}
\end{lemma}
\begin{lemma}
\label{lemma:expGam}
\emph{\citep[Ex.~1.A.24]{Shaked2007}}
Let $X_i\sim\Exp(\lambda_i)$ independently, $i=1,\ldots,m$. Moreover, let $Y_i\sim\Exp(\eta)$ independently, $i=1,\ldots,m$. Then we have
\begin{equation*}
\sum_{i=1}^mX_i\succsim\sum_{i=1}^mY_i\qquad\text{if and only if}\qquad\prod_{i=1}^m\lambda_i\leq\eta^m.
\end{equation*}
\end{lemma}
\begin{lemma}
\label{lemma:S}
Let $S_m$ denote the sum of the $m$ lightest edge weights in $G$. For all $\phi\leq(n-1)/n$ and $c\in [0,2\phi^2/e]$ we have
\begin{equation*}
\mathbb{P}\left(S_{\phi n} \leq \frac{c}{\beta}\right) \leq \exp\left(\phi n\left(2+\ln\left(\frac{c}{2\phi^2}\right)\right)\right).
\end{equation*}
Furthermore, $\mathsf{TSP}\geq\mathsf{MM} \geq S_{n/2}$, where $\mathsf{TSP}$ and $\mathsf{MM}$ are the total distance of a shortest TSP tour and a minimum-distance perfect matching, respectively.
\end{lemma}
\begin{proof}
	Since all edge weights are independent and standard exponential distributed, we have $S_1\sim\Exp(|E|)$. Using the memorylessness property of the exponential distribution, it follows that $S_2-S_1\sim S_1+\Exp(|E|-1)$, i.e., the second lightest edge weight is equal to the lightest edge weight plus the minimum of $|E|-1$ standard exponential distributed random variables. In general, we get $S_{k+1}-S_k\sim S_k-S_{k-1}+\Exp(|E|-k)$. This yields
	\begin{equation*}
	S_{\phi n}\sim\sum_{i=0}^{\phi n-1}(\phi n-i)\cdot\Exp(|E|-i)\sim\sum_{i=0}^{\phi n-1}\Exp\left(\frac{|E|-i}{\phi n-i}\right)\succsim\sum_{i=0}^{\phi n-1}\Exp\left(\frac{e|E|}{\phi n}\right),
	\end{equation*}
	where the stochastic dominance follows from Lemma \ref{lemma:expGam} by observing that
	\begin{equation*}
	\prod_{i=0}^{\phi n-1}\frac{|E|-i}{\phi n-i}=\frac{|E|!}{(\phi n)!(|E|-\phi n)!}=\binom{|E|}{\phi n}\leq\left(\frac{e|E|}{\phi n}\right)^{\phi n},
	\end{equation*} 
	where the inequality follows from applying the well-known inequality $\binom{m}{k}\leq(em/k)^k$.	Next, observe that $|E|\leq\beta n(n-1)/2$. Applying this fact, and then combining it with Lemma \ref{technical:Janson} with $\mu=2\phi^2n/\beta(n-1)$, $a_*=\beta(n-1)/2\phi$ and $\lambda=ec(n-1)/2\phi^2n$ (note that $\lambda\leq1$ since $0\leq c\leq 2\phi^2/e$), we obtain
	\begin{align*}
	\mathbb{P}\left(S_{\phi n}\leq\frac{c}{\beta}\right)&\leq\mathbb{P}\left(\sum_{i=0}^{\phi n-1}\Exp\left(\frac{\beta(n-1)}{2\phi}\right)\leq\frac{ec}{\beta}\right)\\
	&\leq\exp\left(-\phi n\left(\frac{ec(n-1)}{2\phi^2n}-1-\ln\left(\frac{ec(n-1)}{2\phi^2n}\right)\right)\right)\\
	&\leq\exp\left(\phi n\left(2+\ln\left(\frac{c}{2\phi^2}\right)\right)\right).
	\end{align*}
	It remains to show that $\mathsf{TSP}\geq\mathsf{MM} \geq S_{n/2}$. The first inequality follows trivially. For the second one, consider a minimum-distance perfect matching. Take the union of the shortest path between each matched pair of vertices. This union must contain at least $n/2$ different edges of $G$. These edges must have a total weight of at least $S_{n/2}$ and at most $\mathsf{MM}$. So, $\mathsf{MM} \geq S_{n/2}$.\blokje
\end{proof}
\begin{theorem}
\label{theorem:EGR}
The greedy heuristic for minimum-distance perfect matching has an expected approximation ratio on generalized random shortest path metrics given by $\mathbb{E}\left[\frac{\mathsf{GR}}{\mathsf{MM}}\right]=O\left(\beta/\alpha\right)$.
\end{theorem}
\begin{proof}
	Let $c>0$ be a sufficiently small constant. Then the approximation ratio of the greedy heuristic on generalized random shortest path metrics is
	\begin{equation*}
	\mathbb{E}\left[\frac{\mathsf{GR}}{\mathsf{MM}}\right]\leq\mathbb{E}\left[\frac{\beta\cdot\mathsf{GR}}{c}\right]+ \mathbb{E}\left[\frac{\mathsf{GR}}{\mathsf{MM}}\;\middle|\;\mathsf{MM}<\frac{c}{\beta}\right]\cdot\mathbb{P}\left(\mathsf{MM}<\frac{c}{\beta}\right).
	\end{equation*}
	The first term is $O(\beta/\alpha)$ by Theorem \ref{theorem:greedy}. The expectation in the second term can be bounded by the worst-case approximation ratio of the greedy heuristic on metric instances, i.e. $n^{\log_2(3/2)}$ \cite{Reingold1981}. The probability can be bounded by $\exp(\tfrac12n(2+\ln(2c)))$ according to Lemma \ref{lemma:S}. Since $c$ is sufficiently small, this implies that the second term becomes $o(1)$.\blokje
\end{proof}

\paragraph{Nearest Neighbor Heuristic for TSP.}
The nearest-neighbor heuristic is a greedy approach for the TSP: start with some starting vertex $v_0$ as current vertex $v$; at every step, choose the nearest unvisited neighbor $u$ of $v$ as the next vertex in the tour and move to the next iteration with the new vertex $u$ as current vertex $v$; go back to $v_0$ if all vertices are visited. From now on, let $\mathsf{NN}$ denote the cost of the TSP tour computed by this heuristic and let $\mathsf{TSP}$ denote the value of an optimal TSP tour.

The worst-case approximation ratio of this heuristic on metric instances is known to be $O(\ln(n))$~\cite{Rosenkrantz1977}. Furthermore, for random shortest path metrics on complete graphs (for which the cut parameters are given by $\alpha=\beta=1$) the heuristic has an expected approximation ratio of $O(1)$~\cite{Bringmann2015}. We extend this last result to general values for $\alpha$ and $\beta$ and show that the nearest-neighbor heuristic has an expected approximation ratio of $O(\beta/\alpha)$.
\begin{theorem}
\label{theorem:NN}
	For generalized random shortest path metrics, we have $\mathbb{E}[\mathsf{NN}]=O\left(1/\alpha\right)$ and $\mathbb{E}\left[\frac{\mathsf{NN}}{\mathsf{TSP}}\right]=O\left(\beta/\alpha\right)$.
\end{theorem}
\begin{proof}
	The first part of the proof is similar to the proof of Theorem \ref{theorem:greedy}. Let $\Delta_i:=i/\alpha n$. We put the `edges' added to the tour by the nearest-neighbor heuristic into bins depending on their distance, bin $i$ gets the `edges' $\{u,v\}$ with $d(u,v)\in(4\Delta_{i-1},4\Delta_i]$. Using Lemma \ref{lemma:diameter}, we can show that the expected sum of all distances greater than or equal to $\Delta_{\omega(\ln(n))}$ is $o(1/\alpha)$, so we can ignore the corresponding bins in our analysis.
	
	We now estimate the contribution of the other bins to the distance of the TSP tour. By Theorem \ref{theorem:cluster}, we can partition the vertices in an expected number of $O(1+n/\exp((i-1)/5))$ clusters, each of diameter at most $4\Delta_{i-1}$. Every time the nearest-neighbor heuristic adds an `edge' of distance greater than $4\Delta_{i-1}$, this must be an edge from some cluster $C_k$ to another cluster $C_\ell$. Moreover, at this point the partial TSP tour must already have visited all vertices in the cluster $C_k$. Therefore, this can happen at most $O(1+n/\exp((i-1)/5))$ times in expectation. Therefore, bin $i$ can get at most $O(1+n/\exp((i-1)/5))$ `edges' during the run of the nearest-neighbor heuristic. So, the total contribution of bin $i$ is in expectation at most $O(\frac{i}{\alpha n}(1+n/\exp((i-1)/5)))$. Summing over all bins yields
	\begin{equation*}
	\mathbb{E}[\mathsf{NN}]=o\left(\frac1\alpha\right)+\sum_{i=1}^{O(\ln(n))}O\left(\frac1\alpha\left(\frac{i}{n}+\frac{i}{e^{(i-1)/5}}\right)\right)= o\left(\frac1\alpha\right)+O\left(\frac1\alpha\right)=O\left(\frac1\alpha\right).
	\end{equation*}
	Using the worst-case approximation ratio of the nearest-neighbor heuristic on metric instances of $O(\ln(n))$~\cite{Rosenkrantz1977}, the proof for the expected approximation ratio is analogously to the proof of Theorem \ref{theorem:EGR}.\blokje
\end{proof}

\paragraph{Insertion Heuristics for TSP.}
The insertion heuristics are another greedy approach for the TSP: start with an initial optimal tour on a few vertices chosen according to some predefined rule $R$; at every step, choose a vertex according to the same predefined rule $R$ and insert this vertex in the current tour such that the total distance increases the least. From now on, let $\mathsf{IN}_R$ denote the cost of the TSP tour computed by this heuristic (with rule $R$) and let $\mathsf{TSP}$ still denote the value of an optimal TSP tour.

The worst-case approximation ratio of this heuristic for any rule $R$ on metric instances is known to be $O(\ln(n))$~\cite{Rosenkrantz1977}. Furthermore, for random shortest path metrics on complete graphs (for which the cut parameters are given by $\alpha=\beta=1$) the heuristic has an expected approximation ratio of $O(1)$~\cite{Bringmann2015}. We extend this last result to general values for $\alpha$ and $\beta$ and show that the insertion heuristic for any rule $R$ has an expected approximation ratio of $O(\beta/\alpha)$.
\begin{theorem}
\label{theorem:IN}
	For generalized random shortest path metrics, we have $\mathbb{E}[\mathsf{IN}_R]=O\left(1/\alpha\right)$ and $\mathbb{E}\left[\frac{\mathsf{IN}_R}{\mathsf{TSP}}\right]=O\left(\beta/\alpha\right)$.
\end{theorem}
\begin{proof}
	The first part of the proof is similar to the proof of Theorem \ref{theorem:greedy}. Let $\Delta_i:=i/\alpha n$. We put the vertices inserted into the tour by the insertion heuristic into bins depending on the distance they add to the TSP tour, bin $i$ gets the vertices with contribution in the range $(8\Delta_{i-1},8\Delta_i]$. Using Lemma \ref{lemma:diameter}, we can show that the expected sum of all distances greater than or equal to $\Delta_{\omega(\ln(n))}$ is $o(1/\alpha)$, so we can ignore the corresponding bins in our analysis.
	
	We now estimate the contribution of the other bins to the distance of the TSP tour. By Theorem \ref{theorem:cluster}, we can partition the vertices in an expected number of $O(1+n/\exp((i-1)/5))$ clusters, each of diameter at most $4\Delta_{i-1}$. Every time the insertion heuristics adds a vertex that contributes more than $8\Delta_{i-1}$, this must be a vertex that is part of a cluster that is not part of the tour yet. Therefore, this can happen at most $O(1+n/\exp((i-1)/5))$ times in expectation. Therefore, bin $i$ can get at most $O(1+n/\exp((i-1)/5))$ vertices during the run of the insertion heuristic. So, the total contribution of bin $i$ is in expectation at most $O(\frac{i}{\alpha n}(1+n/\exp((i-1)/5)))$. Summing over all bins, and adding the contribution of the initial tour $T_R$ yields
	\begin{equation*}
	\mathbb{E}[\mathsf{IN}_R]=\mathbb{E}[T_R]+o\left(\frac1\alpha\right)+\sum_{i=1}^{O(\ln(n))}O\left(\frac1\alpha\left(\frac{i}{n}+\frac{i}{e^{(i-1)/5}}\right)\right)=O\left(\frac1\alpha\right),
	\end{equation*}
	since we can use Theorem \ref{theorem:NN} to bound the expected length of the initial tour by $\mathbb{E}[T_R]\leq\mathbb{E}[\mathsf{TSP}]\leq\mathbb{E}[\mathsf{NN}]=O(1/\alpha)$.
	Using the worst-case approximation ratio of the insertion heuristic for any rule $R$ on metric instances of $O(\ln(n))$~\cite{Rosenkrantz1977}, the proof for the expected approximation ratio is analogously to the proof of Theorem \ref{theorem:EGR}. Note that this entire proof is independent of the rule $R$ used.\blokje
\end{proof}

\paragraph{Running Time of 2-opt Heuristic for TSP.}
The 2-opt heuristic is an often used local search algorithm for the TSP: start with an initial tour on all vertices and improve the tour by 2-exchanges until no improvement can be made anymore. In a 2-exchange, the heuristic takes `edges' $\{v_1,v_2\}$ and $\{v_3,v_4\}$, where $v_1$, $v_2$, $v_3$, $v_4$ are visited in this order in the tour, and replaces them by $\{v_1,v_3\}$ and $\{v_2,v_4\}$ to create a shorter tour.


We provide an upper bound for the expected number of iterations that 2-opt needs. In the worst-case scenario, this number is exponential. However, for random shortest path metrics on complete graphs (for which the cut parameters are given by $\alpha=\beta=1$) an upper bound of $O(n^8\ln^3(n))$ is known for the expected number of iterations~\cite{Bringmann2015}. We extend this result with a similar proof to general values for $\alpha$ and $\beta$ and show an upper bound for the expected number of iterations of $O(n^8\ln^3(n)\beta/\alpha)$.

We first define the improvement obtained from a 2-exchange. If $\{v_1,v_2\}$ and $\{v_3,v_4\}$ are replaced by $\{v_1,v_3\}$ and $\{v_2,v_4\}$, then the improvement made by the exchange equals the change in distance $\zeta =d(v_1,v_2)+d(v_3,v_4) - d(v_1,v_3)-d(v_2,v_4)$. These four distances correspond to four shortest paths ($P_{12}$, $P_{34}$, $P_{13}$, $P_{24}$) in the graph $G=(V,E)$. This implies that we can rewrite $\zeta$ as the sum of the weights on these paths. We obtain $\zeta =\sum_{e \in E} \gamma_e w(e)$, for some $\gamma_e \in \{-2, -1, 0, 1, 2\}$.

Since we are looking at the improvement obtained by a 2-exchange, we have $\zeta >0$. This implies that there exists some $e=\{u,u'\}\in E$ such that $\gamma_e\neq 0 $. Given this edge $e$, let $I \subseteq \{P_{12},P_{34},P_{13},P_{24}\}$ be the set of all shortest paths of the 2-exchange that contain $e$. Then, for all combinations $e$ and $I$, let $\zeta_{ij}^{e,I}$ be defined as follows:
\begin{itemize}
	\item If $P_{ij} \notin I$, then $\zeta_{ij}^{e,I}$ is the length of the shortest path from $v_i$ to $v_j$ without using $e$.
	\item If $P_{ij} \in I$, then $\zeta_{ij}^{e,I}$ is the minimum of
		\begin{itemize}
			\item the length of a shortest path from $v_i$ to $u$ without using $e$ plus the length of a shortest path from $u'$ to $v_j$ without using $e$ and
			\item the length of a shortest path from $v_i$ to $u'$ without using $e$ plus the length of a shortest path from $u$ to $v_j$ without using $e$.
		\end{itemize}
\end{itemize}
Define $\zeta^{e,I} = \zeta_{12}^{e,I} +  \zeta_{34}^{e,I} - \zeta_{13}^{e,I} -  \zeta_{24}^{e,I}$.
\begin{lemma}
\label{lemma:zeta1}
For every outcome of the edge weights, there exists an edge $e$ and a set $I$ such that $\zeta = \zeta^{e,I} + \gamma w(e)$, where $\gamma \in\{-2,-1,1,2\}$ is determined by $e$ and $I$.
\end{lemma}
\begin{proof}
	Fix the edge weights arbitrarily and consider the four shortest paths from the 2-exchange. As previously stated there exists some edge $e$ with non-zero value $\gamma_e$. Choose this $e$, the corresponding set $I$ and take $\gamma = \gamma_e$. Then the result follows from the definition of $\zeta^{e,I}$.\blokje
\end{proof}
\begin{lemma}
\label{lemma:zeta2}
Let $e$ and $I$ be given with $\gamma = \gamma_e \neq 0$. Then $\mathbb{P}(\zeta^{e,I} +\gamma w(e) \in (0,x])\leq x$. Moreover, $\mathbb{P}(\zeta\in(0,x])=O(\beta n^2x)$.
\end{lemma}
\begin{proof}
	Fix all edge weights except for $w(e)$. Then the value of $\zeta^{e,I}$ is known. Therefore we have $\zeta^{e,I} + \gamma w(e) \in (0,x]$ if and only if $w(e)$ takes a value in an interval of length $x/|\gamma| \leq x$. The first part of the result follows, since $w(e)$ is drawn from $\Exp(1)$ and the density function of this distribution does not exceed $1$.
	Observe that the number of possible choices for $e$ and $I$ is bounded by $|E|\leq\beta n(n-1)/2=O(\beta n^2)$. The second part of the result follows now using Lemma \ref{lemma:zeta1} and a union bound.\blokje
\end{proof}
\begin{theorem}
\label{theorem:2-opt}
The expected number of iterations of the 2-opt heuristic until a local optimum is found is bounded by $O(n^8\ln^3(n)\beta/\alpha)$.
\end{theorem}
\begin{proof}
	Let $\zeta_{\min}>0$ be the minimum improvement that can be made by any 2-exchange. The total number of different 2-exchanges is $O(n^4)$, so using Lemma \ref{lemma:zeta2} and a union bound we obtain $\mathbb{P}(\zeta_{\min}\leq y)=O(\beta n^6y)$.
	
	The initial tour has a length of at most $n\Delta_{\max}$. Let $T$ be the number of iterations taken by the 2-opt heuristic. Then we have $T\leq n\Delta_{\max}/\zeta_{\min}$. So, $T>x$ implies $\Delta_{\max}/\zeta_{\min} >x/n$. This event is contained in the union of the events $\Delta_{\max}> c\ln(x) \ln(n)/\alpha n$ and $\zeta_{\min} < c\ln(x)\ln(n)/\alpha x$, where $c$ is a sufficiently large constant. By Lemma \ref{lemma:diameter} the first event happens with probability at most $n^{2-c\ln(x)/4}=n^{-\Omega(\ln(x))}$. The second event happens with probability at most $O(\beta n^6\ln(n)\ln(x)/\alpha x)$. So, we have
	\begin{equation*}
	\mathbb{P}(T>x)\leq n^{-\Omega(\ln(x))}+O\left(\beta n^6\ln(n)\ln(x)/\alpha x\right).
	\end{equation*}
	The number of iterations is bounded by $n!$, so we obtain
	\begin{equation*}
	\mathbb{E}[T]\leq\sum_{x=1}^{n!}\left(n^{-\Omega(\ln(x))}+O\left(\beta n^6\ln(n)\ln(x)/\alpha x\right)\right).
	\end{equation*}
	The sum of the $n^{-\Omega(\ln(x))}$ contributes a negligible $O(\ln(n!))$. The sum of the remaining $O(\beta n^6\ln(n)\ln(x)/\alpha x)$ contributes $O(\beta n^6\ln(n)\ln^2(n!)/\alpha)=O(n^8\ln^3(n)\beta/\alpha)$.\blokje
\end{proof}

\paragraph{Trivial Heuristic for $k$-Median.}
The goal of the (metric) $k$-median problem is to find a set $U \subseteq V$ of size $k$ such that $\sum_{v \in V} \min_{u\in U} d(v,u)$ is minimized. The best known approximation algorithm for this problem achieves an approximation ratio of $2.675+\varepsilon$~\cite{Byrka2015}.

Here, we consider the $k$-median problem in the setting of generalized random shortest path metrics. We analyze a trivial heuristic for the $k$-median problem: simply pick $k$ vertices independently of the metric space, e.g., $U=\{v_1,\ldots,v_k\}$. The worst-case approximation ratio of this heuristic is unbounded, even if we restrict ourselves to metric instances. However, for random shortest path metrics on complete graphs (for which the cut parameters are given by $\alpha=\beta=1$) the expected approximation ratio has an upper bound of $O(1)$ and even $1+o(1)$ for $k$ sufficiently small~\cite{Bringmann2015}. We extend this result to general values for $\alpha$ and $\beta$ and give an upper bound for the expected approximation ratio of $O(\beta/\alpha)$ for `large' $k$ and $\beta/\alpha + o(\beta/\alpha)$ for $k$ sufficiently small.

For our analysis, let $U=\{v_1,\ldots, v_k\}$ be an arbitrary set of $k$ vertices. Sort the remaining vertices $\{v_{k+1}, \ldots, v_n\}$ in increasing distance from $U$. For $k+1 \leq i\leq n$, let $\rho_{i} = d(v_i,U)$ equal the distance from $U$ to the $(i-k)$-th closest vertex to $U$. Let $\mathsf{TR}$ denote the cost of the solution generated by the trivial heuristic and let $\mathsf{ME}$ be the cost of an optimal solution to the $k$-median problem.

Observe that the random variables $\rho_i$ are generated by a simple growth process analogously to the one described in Section \ref{sect:structural} for $\tau_k(v)$. Using this observation, we can see that
\begin{equation*}
\sum_{j=k}^{i-1}\Exp(\beta j(n-j)) \precsim \rho_i \precsim \sum_{j=k}^{i-1} \Exp(\alpha j(n-j)),
\end{equation*}
which in turn implies that $\mathsf{cost}(U)=\sum_{i=k+1}^n\rho_i$ is stochastically bounded by
\begin{equation*}
\sum_{i=k}^{n-1}\Exp(\beta i) \precsim \mathsf{cost}(U) \precsim \sum_{i=k}^{n-1} \Exp(\alpha i).
\end{equation*}
From this, we can immediately derive bounds for the expected value of the $k$-median returned by the trivial heuristic.
\begin{lemma}
\label{lemma:triviale}
Fix $U\subseteq V$ of size $k$. Then, we have $\mathbb{E}[\mathsf{TR}]=\mathbb{E}[\mathsf{cost}(U)]$ and
\begin{equation*}
\frac{1}{\beta}\left(\ln\left(\frac{n-1}{k-1}\right)-1\right)\leq\mathbb{E}[\mathsf{TR}]\leq\frac{1}{\alpha}\left(\ln\left(\frac{n-1}{k-1}\right)+1\right).
\end{equation*}
\end{lemma}
\begin{proof}
	We have $(H_{n-1}-H_{k-1})/\beta=\sum_{i=k}^{n-1}1/\beta i\leq\mathbb{E}[\mathsf{TR}]\leq\sum_{i=k}^{n-1}1/\alpha i=(H_{n-1}-H_{k-1})/\alpha$. Using $\ln(n)\leq H_n\leq\ln(n)+1$ yields the result.\blokje
\end{proof}
Before we provide our result for the expected approximation ratio of the trivial heuristic, we first provide some tail bounds for the distribution of the optimal $k$-median $\mathsf{ME}$ and the trivial solution $\mathsf{TR}$.
\begin{lemma}
\label{lemma:densityu}
Fix $U\subseteq V$ of size $k$. Then the probability density function $f$ of $\sum_{i=k}^{n-1} \Exp(\beta i)$ is given by
\begin{equation*}
f(x)= \beta k \cdot\binom{n-1}{k} \cdot\exp(-\beta kx)\cdot\left(1-\exp(-\beta x)\right)^{n-k-1}.
\end{equation*}
\end{lemma}
\begin{proof}
	The distribution corresponds to the $(n-k)$-th smallest element out of $n-1$ independent, exponentially distributed random variables with parameter $\beta$. The density of this distribution is known~\cite[Example 2.38]{Ross2010}.\blokje
\end{proof}
\begin{lemma}
\label{lemma:ksmall}
Let $c>0$ be sufficiently large and let $k\leq c'n$ for $c'=c'(c) >0$ sufficiently small. Then we have
\begin{equation*}
\mathbb{P}\left(\mathsf{ME} \leq \left(\ln\left(\tfrac{n-1}{k}\right) - \ln\ln\left(\tfrac{n}{k}\right)-\ln(c)\right)/\beta\right)= n^{-\Omega(c)}.
\end{equation*}
\end{lemma}
\begin{proof}
	We first want a bound for $f(x)$ at $x=\ln((n-1)/ak)/\beta$ for sufficiently large $a$ with $1\leq a\leq (n-1)/k$. For this particular value of $x$, by Lemma \ref{lemma:densityu} we have,
	\begin{equation*}
	f(x)=\beta k \cdot\binom{n-1}{k}\cdot\frac{(ak)^k(n-1-ak)^{n-k-1}}{(n-1)^{n-1}}\leq\beta k (ae)^k\left(1-\frac{ak}{n-1}\right)^{n-k-1},
	\end{equation*}
	where we used $\binom{n-1}{k} \leq ((n-1)e/k)^k$ for the inequality. Since $1+x\leq\exp(x)$ and $(n-k-1)/(n-1)=\Omega(1)$ (since $k$ is sufficiently small), we obtain
	\begin{equation*}
	f(x)\leq\beta k(ae)^k\exp(-\Omega(ak)).
	\end{equation*}
	Since $a$ is sufficiently large, the first factors (without the $\beta$) are lower order terms that can be hidden by the $\Omega$. This implies that $f(x)\leq\beta\exp(-\Omega(ak))$. Substituting $a=(n-1)\exp(-\beta x)/k$ into this yields
	\begin{equation*}
	f(x)\leq\beta\exp(-\Omega((n-1)\exp(-\beta x))),
	\end{equation*}
	which holds for $x\in[0,\ln((n-1)/bk)/\beta]$ for $b \geq 1$ sufficiently large. Recall that $\mathsf{cost}(U)\succsim\sum_{i=k}^{n-1}\Exp(\beta i)$. So, we have $\mathbb{P}(\mathsf{cost}(U)<\ln((n-1)/bk)/\beta)\leq\mathbb{P}(\sum_{i=k}^{n-1}\Exp(\beta i)<\ln((n-1)/bk)/\beta)$. This latter probability is equal to
	\begin{align*}
	\int_{0}^{\ln\left(\frac{n-1}{bk}\right)/\beta} f(x)\,\mathrm{d}x &= \int_{0}^{\ln\left(\frac{n-1}{bk}\right)/\beta} f\left(\ln\left(\tfrac{n-1}{bk}\right)/\beta-x\right)\,\mathrm{d}x\\
	&\leq \int_{0}^{\ln\left(\frac{n-1}{bk}\right)/\beta} \beta\exp\left(-\Omega(bk\exp(\beta x))\right)\,\mathrm{d}x\\
	&\leq \int_{0}^{\ln\left(\frac{n-1}{bk}\right)} \exp\left(-\Omega(bk\exp(x))\right)\,\mathrm{d}x\\
	&\leq \int_{0}^{\infty}\exp\left(-\Omega(bk(1+x))\right)\,\mathrm{d}x\leq\exp(-\Omega(bk)),
	\end{align*}
	where the last step follows from the fact that $\int_{0}^{\infty} \exp(-\Omega(bkx))\,\mathrm{d}x = O(1/bk) \leq 1$ as $b$ is sufficiently large.
	
	In order for $\mathsf{ME}$ to be small, there must exist a subset $U \subseteq V$ of size $k$ that has low cost. We bound this probability by taking a union bound, which yields
	\begin{align*}
	\mathbb{P}\left(\mathsf{ME} < \ln\left(\tfrac{n-1}{bk}\right)/\beta\right) &= \mathbb{P}\left(\exists\, U \subseteq V,\, |U|=k:\mathsf{cost}(U) <\ln\left(\tfrac{n-1}{bk}\right)/\beta\right)\\
	&\leq\binom{n}{k} \cdot \mathbb{P}\left(\mathsf{cost}(U) <\ln\left(\tfrac{n-1}{bk}\right)/\beta\right)\\
	&\leq\binom{n}{k} \cdot \exp(-\Omega(bk)).
	\end{align*}
	Set $b=c\ln(n/k)$ for sufficiently large $c\geq 1$. Then we fulfill the condition that $b\geq 1$ and sufficiently large. Combining this with $\binom{n}{k}\leq(ne/k)^k$ yields
	\begin{equation*}
	\mathbb{P}\left(\mathsf{ME} < \left(\ln\left(\tfrac{n-1}{k}\right) -\ln\ln\left(\tfrac{n}{k}\right)-\ln(c)\right)/\beta\right)\leq\left(\frac{en}{k}\right)\cdot\left(\frac{n}{k}\right)^{-\Omega(ck)}.
	\end{equation*}
	Since $k$ is sufficiently smaller than $n$, we have $en/k \leq (n/k)^2$. As $c$ is sufficiently large, we can simplify the right hand side to $(n/k)^{-\Omega(ck)}$. Finally, since $k\geq 1$ and $k$ is sufficiently smaller than $n$, we have $(n/k)^k\geq n$. This implies $(n/k)^{-\Omega(ck)}\leq n^{-\Omega(c)}$, which competes the proof.\blokje
\end{proof}
\begin{lemma}
\label{lemma:klarge}
Let $k\leq (1-\varepsilon) n$ for some constant $\varepsilon >0$. For every $c\in [0,2\varepsilon^2)$, we have
\begin{equation*}
\mathbb{P}\left(\mathsf{ME} \leq  c/\beta \right)\leq c^{\Omega(n)}.
\end{equation*}
\end{lemma}
\begin{proof}
	The value of $\mathsf{ME}$ is the sum of $n-k$ shortest path lengths in $G$. The union of these paths contains at least $n-k$ different edges from $G$. Let $S_m$ be the sum of the $m$ lightest edge weights in $G$. We obtain $\mathsf{ME} \geq S_{n-k}\geq S_{\varepsilon n}$. The result follows using Lemma \ref{lemma:S} with $\phi=\varepsilon$.\blokje
\end{proof}
\begin{lemma}
\label{lemma:trivialp}
For any $c\geq 4$ we have $\mathbb{P}\left(\mathsf{TR} > n^c\right) \leq \exp(-n^{c/4})$.
\end{lemma}
\begin{proof}
	We can roughly bound $\mathsf{TR}$ by $n\Delta_{\max}$, which in turn can be roughly bounded by $n^2\max_e\{w(e)\}$. Since $\max_e\{w(e)\}$ is the maximum of $|E|\leq\beta n(n-1)/2$ independent exponentially distributed random variables with parameter $1$, we have
	\begin{align*}
	\mathbb{P}\left(\mathsf{TR} \leq n^c\right) &\geq \left(1-\exp\left(-n^{c-2}\right)\right)^{\beta n(n-1)/2}\geq 1- \tfrac12\beta n(n-1) \cdot \exp\left(-n^{c-2}\right)\\
	&\geq 1-\exp\left(-n^{c-3}\right) \geq 1-\exp\left(-n^{c/4}\right).
	\end{align*}
	The result follows by taking the complement.\blokje
\end{proof}
Now we have obtained everything needed to provide an upper bound for the expected approximation ratio of the trivial heuristic.
\begin{theorem}
\label{theorem:trivial}
Let $k\leq (1-\varepsilon)n$ for some constant $\varepsilon>0$. For generalized random shortest path metrics, we have $\mathbb{E}\left[\frac{\mathsf{TR}}{\mathsf{ME}}\right] = O\left(\beta/\alpha\right)$. Moreover, if we have $k\leq c'n$ for some fixed $c'\in(0,1)$ sufficiently small, then we have
\begin{equation*}
\mathbb{E}\left[\tfrac{\mathsf{TR}}{\mathsf{ME}}\right] = (\beta/\alpha)\cdot\left(1+O\left(\tfrac{\ln\ln(n/k)}{\ln(n/k)}\right)\right).
\end{equation*}
\end{theorem}
\begin{proof}
	We have for all constants $m>0$
	\begin{equation*}
	\mathbb{E}\left[\frac{\mathsf{TR}}{\mathsf{ME}}\right]  \leq \mathbb{E}\left[\frac{\beta\cdot\mathsf{TR}}{m}\right] +\mathbb{P}\left(\mathsf{ME}<\frac{m}{\beta}\right)\cdot \mathbb{E}\left[\frac{\mathsf{TR}}{\mathsf{ME}}\;\middle|\; \mathsf{ME} < \frac{m}{\beta}\right].
	\end{equation*}
	\emph{Case 1 ($k\leq c'n$, $c'$ sufficiently small)}: Let $n$ be sufficiently large. According to Lemma \ref{lemma:ksmall} we can pick a constant $c>0$ sufficiently large such that
	\begin{equation*}
	\mathbb{P}\left(\mathsf{ME} \leq  \left(\ln\left(\tfrac{n-1}{k}\right) - \ln\ln\left(\tfrac{n}{k}\right)-\ln(c)\right)/\beta\right) \leq n^{-9}.
	\end{equation*}
	Take $m=\ln((n-1)/k)-\ln\ln(n/k)-\ln(c)$. By Lemma~\ref{lemma:triviale}, we have
	\begin{equation*}
	\mathbb{E}\left[\frac{\beta\cdot\mathsf{TR}}{m}\right] \leq \frac\beta\alpha\cdot\frac{\ln\left(\frac{n-1}{k-1}\right)+1}{m}\leq\frac\beta\alpha\cdot\left(1+O\left(\frac{\ln\ln(n/k)}{\ln(n/k)}\right)\right).
	\end{equation*}
	For the second part we can use the fact that $m$ was chosen such that $\mathbb{P}(\mathsf{ME}\leq m/\beta)\leq n^{-9}$ to obtain
	\begin{align*}
	\mathbb{P}\left(\mathsf{ME}<\frac{m}{\beta}\right)\cdot \mathbb{E}\left[\frac{\mathsf{TR}}{\mathsf{ME}}\;\middle|\; \mathsf{ME} < \frac{m}{\beta}\right] &= \mathbb{P}\left(\mathsf{ME}<\frac{m}{\beta}\right) \cdot \int_0^\infty \mathbb{P}\left(\frac{\mathsf{TR}}{\mathsf{ME}} \geq x\;\middle|\; \mathsf{ME} < \frac{m}{\beta}\right) \mathrm{d}x\\
	&\leq \mathbb{P}\left(\mathsf{ME}<\frac{m}{\beta}\right) \cdot\left(n^8+ \int_{n^8}^\infty \mathbb{P}\left(\frac{\mathsf{TR}}{\mathsf{ME}} \geq x\;\middle|\; \mathsf{ME} < \frac{m}{\beta}\right) \mathrm{d}x\right)\\
	&\leq \frac1n+ \int_{n^8}^\infty \mathbb{P}\left(\frac{\mathsf{TR}}{\mathsf{ME}} \geq x\; \text{and} \; \mathsf{ME} < \frac{m}{\beta}\right) \mathrm{d}x\\
	&\leq \frac1n+ \int_{n^8}^\infty \mathbb{P}\left(\frac{\mathsf{TR}}{\mathsf{ME}} \geq x\right) \mathrm{d}x\\
	&\leq \frac1n+ \int_{n^8}^\infty\mathbb{P}\left(\mathsf{TR}\geq \sqrt{x}\right) \mathrm{d}x+\int_{n^8}^\infty \mathbb{P}\left(\mathsf{ME}\leq \frac{1}{\beta \sqrt{x}}\right) \mathrm{d}x,
	\intertext{where the last inequality follows since $\mathsf{TR}/\mathsf{ME} \geq x$ implies $\mathsf{TR} \geq \sqrt{x}$ or $\mathsf{ME} \leq 1/\sqrt{x} \leq 1/\beta \sqrt{x}$. Note that the requirements for applying Lemmas \ref{lemma:klarge} and \ref{lemma:trivialp} to the corresponding probabilities are met for any $x\in[n^8,\infty)$. Upon applying those we obtain}
	\mathbb{P}\left(\mathsf{ME}<\frac{m}{\beta}\right)\cdot \mathbb{E}\left[\frac{\mathsf{TR}}{\mathsf{ME}}\;\middle|\; \mathsf{ME} < \frac{m}{\beta}\right] &\leq \frac1n+ \int_{n^8}^\infty\exp\left(-x^{1/8}\right) \mathrm{d}x+\int_{n^8}^\infty\left(\frac{1}{\sqrt{x}}\right)^{\Omega(n)}\mathrm{d}x\\
	&= O\left(\frac{1}{n}\right).
	\end{align*}
	\emph{Case 2 ($c'n<k\leq (1-\varepsilon)n$, $\varepsilon>0$)}: We repeat the proof for the previous case, but this time we choose $m$ as a sufficiently small constant ($m<\min\{2\varepsilon^2,1\}$ satisfies). Then, by Lemma~\ref{lemma:klarge}, we have $\mathbb{P}\left(\mathsf{ME}<m/\beta\right) \leq m^{\Omega(n)} \leq n^{-9}$. Furthermore, by Lemma \ref{lemma:triviale}, we have
	\begin{equation*}
	\mathbb{E}\left[\frac{\beta\cdot\mathsf{TR}}{m}\right] \leq \frac\beta\alpha\cdot \frac{\ln\left(\frac{n-1}{k-1}\right) + 1}{m}=O\left(\frac\beta\alpha\right),
	\end{equation*}
	since $k>c'n$. Together with the second part of the first case, this shows the claim.\blokje
\end{proof}

\section{Application to the Erd\H{o}s-R\'enyi Random Graph Model}\label{sect:ERRG}

So far, we have analyzed random shortest path metrics applied to graphs based on their cut parameters (Def.~\ref{def:ab}). In this section, we first show that instances of the Erd\H{o}s--R\'enyi random graph model have `nice' cut parameters with high probability. We then use this to prove our main results.
\begin{lemma}
	\label{lemma:ep}
	Let $G=(V,E)$ be an instance of the $G(n,p)$ model. For constant $\varepsilon \in (0,1)$ and for any $p\geq c\ln(n)/n$ (as $n\to\infty$), in which $c>9/\varepsilon^2$ is constant, the cut parameters of $G$ are bounded by $(1-\varepsilon)p\leq\alpha\leq\beta\leq(1+\varepsilon)p$ with probability at least $1-o\left(1/n^2\right)$.
\end{lemma}
\begin{proof}
	Let $\mathcal{E}$ denote the event that the cut parameters of $G$ are not bounded by $(1-\varepsilon)p\leq\alpha\leq\beta\leq(1+\varepsilon)p$. Using the definition of the cut parameters, the probability of this event can be written as
	\begin{equation*}
	\mathbb{P}(\mathcal{E})=\mathbb{P}\left(\exists\;\varnothing\neq U\subset V,\,|U|\leq n/2:\big||\delta(U)|-p\mu_U\big|>\varepsilon p\mu_U\right).
	\end{equation*}
	We can restrict ourselves here to subsets $U$ of size at most $n/2$ since $U$ and $V\backslash U$ induce the same cut of $G$. Using the union bound, we can bound this probability by
	\begin{equation*}
	\mathbb{P}(\mathcal{E})\leq\sum_{k=1}^{n/2}\binom{n}{k}\cdot\mathbb{P}\left(\big||\delta(U_k)|-p\mu_U\big|>\varepsilon p\mu_U\right),
	\end{equation*}
	where $U_k$ is a subset of $V$ of size $k$. Applying a Chernoff bound~\cite[Cor.~4.6]{Mitzenmacher2005} to each term of this summation, we can further bound this by
	\begin{align*}
	\mathbb{P}(\mathcal{E})&\leq\sum_{k=1}^{n/2}\binom{n}{k}\cdot2e^{-k(n-k)p\varepsilon^2/3}\\
	&\leq\sum_{k=1}^{n/2}\binom{n}{k}\cdot2e^{-k(n-k)c\ln(n)\varepsilon^2/3n},
	\end{align*}
	where we used $p\geq c\ln(n)/n$ for the last inequality. Now, let $\xi>0$ be sufficiently small ($\xi<1-9/c\varepsilon^2$ satisfies). Using this $\xi$, we split the summation in two parts, and use the bounds $\binom{n}{k}\leq n^k$ and $\binom{n}{k}\leq 2^n$, respectively, to obtain
	\begin{align*}
	\mathbb{P}(\mathcal{E})&\leq\sum_{k=1}^{\xi n}\binom{n}{k}\cdot2e^{-k(n-k)c\ln(n)\varepsilon^2/3n}+\sum_{k=\xi n}^{n/2}\binom{n}{k}\cdot2e^{-k(n-k)c\ln(n)\varepsilon^2/3n}\\
	&\leq\sum_{k=1}^{\xi n}2e^{k\ln(n)(1-(1-k/n)c\varepsilon^2/3)}+\sum_{k=\xi n}^{n/2}2e^{n\ln(2)-k(n-k)c\ln(n)\varepsilon^2/3n}\\
	&\leq\sum_{k=1}^{\xi n}2e^{k\ln(n)(1-(1-\xi)c\varepsilon^2/3)}+\sum_{k=\xi n}^{n/2}2e^{n\ln(2)-\xi(1-\xi)cn\ln(n)\varepsilon^2/3}.
	\end{align*}
	For the last inequality we used the fact that $k/n\leq\xi$ for all $1\leq k\leq\xi n$ and that $k(n-k)\geq\xi(1-\xi)n^2$ for all $\xi n\leq k\leq n/2$. Now, since $\xi$ is sufficiently small, we have $1-(1-\xi)c\varepsilon^2/3<-2$ and thus we can bound the first summation by $o(1/n^2)$. Furthermore, as $n\to\infty$, each summand of the second summation is bounded by $e^{-\Omega(n\ln(n))}=n^{-\Omega(n)}$, which allows us to bound the second summation by $n\cdot n^{-\Omega(n)}=n^{-\Omega(n)}$. Together with the bound for the first summation, this yields $\mathbb{P}(\mathcal{E})\leq o(1/n^2)+n^{-\Omega(n)}=o(1/n^2)$. The result now follows by taking the complement of the event $\mathcal{E}$.\blokje
\end{proof}
Recall that from the result of Corollary \ref{corollary:harmonic} we could derive (approximate) bounds for the expected distance $\mathbb{E}[d(u,v)]$ between two arbitrary vertices in a random shortest path metric. Combining this with the result of the foregoing lemma, we can see that, for the case of the application to the Erd\H{o}s--R\'enyi random graph model, w.h.p. over the random graph $\mathbb{E}[d(u,v)]$ is approximately bounded between $\ln(n)/((1+\varepsilon)np)$ and $\ln(n)/((1-\varepsilon)np)$ for any constant $\varepsilon\in(0,1)$. This is in line with the known result $\mathbb{E}[d(u,v)]\approx\ln(n)/np$ for $p$ sufficiently large~\cite{Bhamidi2011}.

\subsection{Performance of Heuristics}

In this section, we provide the main results of this work. We use the results from Section \ref{sect:heuristics} and Lemma \ref{lemma:ep} to analyze the performance of several heuristics in random shortest path metrics applied to Erd\H{o}s--R\'enyi random graphs.

When a graph $G=(V,E)$ is created by the $G(n,p)$ model, there is a non-zero probability of $G$ being disconnected. In a corresponding random shortest path metric this results in $d(u,v)=\infty$ for any two vertices $u,v \in V$ that are in different components of $G$. Observe that, if this is the case, then the identity of indiscernibles, symmetry and triangle inequality still hold. Thus we still have a metric and we can bound the expected approximation ratio for such graphs from above by the worst-case approximation ratio for metric instances.

Using this observation, we can prove the following results.
\begin{theorem}
\label{theorem:PGR}
Let $\varepsilon\in(0,1)$ be constant. Let $G=(V,E)$ be a random instance of the $G(n,p)$ model, for $p$ sufficiently large ($p\geq c\ln(n)/n$ as $n\to\infty$ for a constant $c>9/\varepsilon^2$ satisfies), and consider the corresponding random shortest path metric. Then, we have
\begin{equation*}
	\mathbb{E}\left[\frac{\mathsf{GR}}{\mathsf{MM}}\right] = O(1).
\end{equation*}
\end{theorem}
\begin{proof}
Let $\mathcal{E}$ denote the event that the cut parameters of $G$ are bounded by $(1-\varepsilon)p\leq\alpha\leq\beta\leq(1+\varepsilon)p$. Then we have
\begin{align*}
\mathbb{E}\left[\frac{\mathsf{GR}}{\mathsf{MM}}\right] &\leq  \mathbb{E}\left[\frac{\mathsf{GR}}{\mathsf{MM}}\;\middle|\; \mathcal{E}\right] +  \mathbb{E}\left[\frac{\mathsf{GR}}{\mathsf{MM}}\;\middle|\; \overline{\mathcal{E}}\;\right] \cdot \mathbb{P}\left(\;\overline{\mathcal{E}}\;\right).\\
&\leq O\left(\frac{(1+\varepsilon)p}{(1-\varepsilon)p}\right)+O\left(n^{\log_2(3/2)}\right)\cdot o\left(\frac1{n^2}\right)=O(1),
\end{align*}
where we used the results of Theorem \ref{theorem:EGR}, Lemma \ref{lemma:ep}, and the worst-case approximation ratio of the greedy heuristic on metric instances~\cite{Reingold1981}.\blokje
\end{proof}
\begin{theorem}
\label{theorem:PNNIN}
Let $\varepsilon\in(0,1)$ be constant. Let $G=(V,E)$ be a random instance of the $G(n,p)$ model, for $p$ sufficiently large ($p\geq c\ln(n)/n$ as $n\to\infty$ for a constant $c>9/\varepsilon^2$ satisfies), and consider the corresponding random shortest path metric. Then, we have
\begin{equation*}
	\mathbb{E}\left[\frac{\mathsf{NN}}{\mathsf{TSP}}\right] = O(1)\qquad\text{and}\qquad\mathbb{E}\left[\frac{\mathsf{IN}_R}{\mathsf{TSP}}\right] = O(1).
\end{equation*}
\end{theorem}
\begin{proof}
	Let $\mathcal{E}$ denote the event that the cut parameters of $G$ are bounded by $(1-\varepsilon)p\leq\alpha\leq\beta\leq(1+\varepsilon)p$. Then we have
	\begin{align*}
	\mathbb{E}\left[\frac{\mathsf{NN}}{\mathsf{TSP}}\right] &\leq  \mathbb{E}\left[\frac{\mathsf{NN}}{\mathsf{TSP}}\;\middle|\; \mathcal{E}\right] +  \mathbb{E}\left[\frac{\mathsf{NN}}{\mathsf{TSP}}\;\middle|\; \overline{\mathcal{E}}\;\right] \cdot \mathbb{P}\left(\;\overline{\mathcal{E}}\;\right).\\
	&\leq O\left(\frac{(1+\varepsilon)p}{(1-\varepsilon)p}\right)+O\left(\ln(n)\right)\cdot o\left(\frac1{n^2}\right)=O(1),
	\end{align*}
	where we used the results of Theorem \ref{theorem:NN}, Lemma \ref{lemma:ep}, and the worst-case approximation ratio of the nearest-neighbor heuristic on metric instances~\cite{Rosenkrantz1977}. For the second part, we use the same argument, which follows this time from the results of Theorem \ref{theorem:IN}, Lemma \ref{lemma:ep}, and the worst-case approximation ratio of the insertion heuristics on metric instances~\cite{Rosenkrantz1977}. Note that this argument is independent of the rule $R$ used.\blokje
\end{proof}
For the last two results, we need the assumption that $G$ is connected.
\begin{theorem}
\label{theorem:P2opt}
Let $\varepsilon\in(0,1)$ be constant. Let $G=(V,E)$ be a random instance of the $G(n,p)$ model, for $p$ sufficiently large ($p\geq c\ln(n)/n$ as $n\to\infty$ for a constant $c>9/\varepsilon^2$ satisfies), and consider the corresponding random shortest path metric. If $G$ is connected, then the expected number of iterations of the 2-opt heuristic for TSP is bounded by $O(n^8\ln^3(n))$.
\end{theorem}
\begin{proof}
	Let $T$ be the number of iterations of the 2-opt heuristic and let $\mathcal{E}$ denote the event that the cut parameters of $G$ are bounded by $(1-\varepsilon)p\leq\alpha\leq\beta\leq(1+\varepsilon)p$, whereas $\mathcal{E}'$ denotes the event that $G$ is connected. Note that $\mathcal{E}$ implies $\mathcal{E}'$. Moreover, note that event $\mathcal{E}'$ implies that the cut parameters of $G$ are bounded by $\Theta(1/n^2)\leq\alpha\leq\beta\leq1$. Now, we have
	\begin{align*}
	\mathbb{E}\left[T\mid\mathcal{E}'\right]&\leq\mathbb{E}\left[T\mid\mathcal{E}',\mathcal{E}\right]+ \mathbb{E}\left[T\mid\mathcal{E}',\overline{\mathcal{E}}\;\right]\cdot\mathbb{P}\left(\;\overline{\mathcal{E}}\;\right)\\
	&\leq O\left(n^8\ln^3(n)\cdot\frac{(1+\varepsilon)p}{(1-\varepsilon)p}\right)+O\left(n^8\ln^3(n)\cdot\frac1{1/n^2}\right)\cdot o\left(\frac1{n^2}\right)\\
	&=O(n^8\ln^3(n)),
	\end{align*}
	where we used the results of Theorem \ref{theorem:2-opt} and Lemma \ref{lemma:ep}.\blokje
\end{proof}
\begin{theorem}
\label{theorem:Ptrivial}
Let $\tilde\varepsilon\in(0,1)$ be constant. Let $G=(V,E)$ be a random instance of the $G(n,p)$ model, for $p$ sufficiently large ($p\geq c\ln(n)/n$ as $n\to\infty$ for a constant $c>9/\tilde\varepsilon^2$ satisfies), and consider the corresponding random shortest path metric. Let $\mathcal{E}'$ denotes the event that $G$ is connected. Let $k\leq (1-\varepsilon')n$ for some constant $\varepsilon'>0$, then we have $\mathbb{E}\left[\frac{\mathsf{TR}}{\mathsf{ME}}\;\middle| \; \mathcal{E}'\right] = O\left(1\right)$. Moreover, if we have $k\leq c'n$ for $c'\in(0,1)$ sufficiently small, then $\mathbb{E}\left[\frac{\mathsf{TR}}{\mathsf{ME}}\;\middle| \; \mathcal{E}'\right] =1+\varepsilon+o(1)$.
\end{theorem}
\begin{proof}
	Let $\mathcal{E}$ denote the event that the cut parameters of $G$ are bounded by $(1-\tilde\varepsilon)p\leq\alpha\leq\beta\leq(1+\tilde\varepsilon)p$. Note that $\mathcal{E}$ implies $\mathcal{E}'$. Moreover, note that event $\mathcal{E}'$ implies that the cut parameters of $G$ are bounded by $\Theta(1/n^2)\leq\alpha\leq\beta\leq1$. Now, we have
	\begin{align*}
	\mathbb{E}\left[\frac{\mathsf{TR}}{\mathsf{ME}}\;\middle|\;\mathcal{E}'\right] &\leq  \mathbb{E}\left[\frac{\mathsf{TR}}{\mathsf{ME}}\;\middle|\; \mathcal{E}',\mathcal{E}\right] +  \mathbb{E}\left[\frac{\mathsf{TR}}{\mathsf{ME}}\;\middle|\; \mathcal{E}',\overline{\mathcal{E}}\;\right] \cdot \mathbb{P}\left(\;\overline{\mathcal{E}}\;\right).\\
	&\leq O\left(\frac{(1+\tilde\varepsilon)p}{(1-\tilde\varepsilon)p}\right)+O\left(n^2\right)\cdot o\left(\frac1{n^2}\right)=O(1),
	\end{align*}
	where we used the results of Theorem \ref{theorem:trivial} and Lemma \ref{lemma:ep}. Moreover, if $k\leq c'n$ for $c'\in(0,1)$ sufficiently small, then
	\begin{align*}
	\mathbb{E}\left[\frac{\mathsf{TR}}{\mathsf{ME}}\;\middle|\;\mathcal{E}'\right] &\leq  \mathbb{E}\left[\frac{\mathsf{TR}}{\mathsf{ME}}\;\middle|\; \mathcal{E}',\mathcal{E}\right] +  \mathbb{E}\left[\frac{\mathsf{TR}}{\mathsf{ME}}\;\middle|\; \mathcal{E}',\overline{\mathcal{E}}\;\right] \cdot \mathbb{P}\left(\;\overline{\mathcal{E}}\;\right).\\
	&\leq O\left(\frac{(1+\tilde\varepsilon)p}{(1-\tilde\varepsilon)p}\cdot\left(1+\frac{\ln\ln(n/k)}{\ln(n/k)}\right)\right)+O\left(n^2\cdot\left(1+\frac{\ln\ln(n/k)}{\ln(n/k)}\right)\right)\cdot o\left(\frac1{n^2}\right)\\
	&=1+\varepsilon+o(1),
	\end{align*}
	where $\varepsilon=(1+\tilde\varepsilon)/(1-\tilde\varepsilon)$ can be made arbitrarily small by taking $\tilde\varepsilon$ sufficiently small, and where we again used the results of Theorem \ref{theorem:trivial} and Lemma \ref{lemma:ep}.\blokje
\end{proof}

\section{Concluding Remarks}\label{sect:final}

We have analyzed heuristics for matching, TSP, and $k$-median on random shortest path metrics on Erd\H{o}s--R\'enyi random graphs. However, in particular for constant values of $p$, these graphs are still dense. Although our results hold for decreasing $p = \Omega(\ln n/n)$, we obtain in this way metrics with unbounded doubling dimension. In order to get an even more realistic model for random metric spaces, it would be desirable to analyze heuristics on random shortest path metrics on sparse graphs. Hence, we raise the question to generalize our findings to sparse random graphs or sparse (deterministic) classes of graphs.

\bibliographystyle{abbrvnat}
\bibliography{References}

\end{document}